\documentclass[12pt]{article}

\usepackage[ansinew]{inputenc}
\usepackage[UKenglish]{babel}
\usepackage{textcomp}
\usepackage{mathrsfs}
\usepackage{amsfonts}
\usepackage{mathtools}
\usepackage{amssymb}
\usepackage{amsmath}
\usepackage{amsthm}
\usepackage{amscd}
\usepackage{tensor}
\usepackage{stmaryrd}
\usepackage{rotating}
\usepackage[text={6.8in,8.5in},centering]{geometry}
\usepackage{tikz}
\usetikzlibrary{arrows, calc, decorations.markings, positioning}
\usepackage{tikz-cd}
\usepackage{relsize} 
\usepackage[none]{hyphenat}
\usepackage{hyperref}

\usepackage{graphicx}
\graphicspath{ {./Images/} }

\usepackage{hyperref,xcolor}
\usepackage{makeidx}
\definecolor{deepred}{rgb}{0.5,0,0}
\definecolor{deepblue}{rgb}{0,0,0.5}
\definecolor{deepgreen}{rgb}{0,0.5,0}

\hypersetup{
    colorlinks=true,
    linkcolor=deepred,
    filecolor=deepred,      
    urlcolor=deepred,
    citecolor=deepblue,
}

\newcommand{\Int}{\mathbb{Z}}
\newcommand{\Rat}{\mathbb{Q}}
\newcommand{\Real}{\mathbb{R}}
\newcommand{\Field}{\mathbb{F}}
\newcommand{\dimm}{\text{dim}}

\newcommand{\Proj}{\text{pr}}
\newcommand{\Id}{\text{id}}

\newcommand{\Dim}[1]{\text{Dim}(#1)}
\newcommand{\Dr}[1]{\text{Der}(#1)}

\newcommand{\Ker}[1]{\text{ker}(#1)}

\newcommand{\Set}{\textsf{Set}}

\newcommand{\Vect}{\textsf{Vect}}

\newcommand{\Line}{\textsf{Line}}

\newtheorem{prop}{Proposition}
\numberwithin{prop}{section}

\newtheorem{thm}{Theorem}
\numberwithin{thm}{section}

\usepackage[
backend=biber,
style=alphabetic,
sorting=ynt
]{biblatex}
\begin{document}

\title{Dimensioned Algebra:\\ the mathematics of physical quantities}

\author{Carlos Zapata-Carratala}
\date{}

\maketitle

\sloppy

\begin{abstract}
    A rigorous mathematical theory of dimensional analysis, systematically accounting for the use of physical quantities in science and engineering, perhaps surprisingly, was not developed until relatively recently. We claim that this has shaped current mathematical models of theoretical physics, which generally lack any explicit reference to units of measurement, and we propose a novel mathematical framework to alleviate this. Our proposal is a generalization of the usual categories of algebraic structures used to formulate physical theories (groups, rings, vector spaces...), herein dubbed \emph{dimensioned}, that can naturally articulate the structure of physical dimension. Our goal in the present work is not so much to define an algebraic theory of physical quantities -- this has already been done -- but to define a theory of algebra informed by how physical quantities are used in practice. We conclude by studying the dimensioned analogue of Poisson algebras in some detail due to their relevance in Jacobi geometry and classical mechanics. These topics are further explored in sequel papers by the author.
\end{abstract}

\tableofcontents

\newpage

\section{Introduction} \label{intro}

The origin of the concept of physical quantity is necessarily entangled with the (pre-)history of mathematics as a whole. Primitive number systems and units of measurement developed hand in hand throughout millennia \cite{coolidge2012numerosity}, with deep roots in ancient technology \cite{overmann2019tools}, land distribution \cite{duncan1980length} \cite{clagett1989ancient} and accountancy \cite{hodgkin2005history}. It was not until the advent of modern abstract mathematics at the end of the 19th century that we find a clear, definitive divide between quantities, resulting from measurements in physical reality, and numbers, as moving parts of purely formal algebraic frameworks \cite{roche1998mathematics}. The mathematical models of modern science, especially in physics, are heavily influenced by this division: they often rely on a wealth of formal structures -- largely lacking any explicit representation of \emph{measured} physical quantities -- which are then adapted ad hoc to incorporate numerical quantities and units of measurement to account for concrete experimental arrangements. In fact, it is standard practice in most such theories and formalisms to simplify physical quantities and to consider them merely as numerical values, i.e. elements of $\Rat$ or $\Real$, thereby completely omitting any reference to units of measurement.\newline

The study of physical quantities themselves, the so-called \emph{dimensional analysis} -- a term we also adopt -- first became a part of the physical science vernacular during the 19th century with the work of Fourier \cite{baron1822theorie}, with partial precedents in the writings of Descartes and Euler on mechanics \cite{nolte2018galileo}. Dimensional analysis didn't become part of the mainstream scientific discourse until the early 20th century with the works of Buckingham and Rayleigh. For a review of the history of dimensional analysis see \cite{macagno1971historico} \cite[Ch. 3]{zapata2019landscape}.\newline

Although systematic treatments of dimensional analysis were common since its inception \cite{subramanian1985vector} \cite{barenblatt1996scaling}, and despite the numerous philosophical disquisitions that have followed since \cite{kyburg1984theory} \cite{hale2002real} \cite{grozier2020should}, rigorous mathematical axiomatizations of dimensional analysis have only appeared relatively recently. The first efforts in developing a general mathematical theory of physical quantities are due to Hart \cite{hart2012multidimensional} in the 1980s. However, it was not until the recent work of Jany{\v{s}}ka-Modugno-Vitolo \cite{janyvska2007semi} \cite{janyvska2010algebraic} on semi-vector spaces and positive spaces that we find all the standard features of dimensional analysis in a transparent and mathematically rigorous framework. In a parallel development not directly motivated by dimensional analysis, Dolan-Baez found a characterisation of physical quantities as line objects in the context of category theory and algebraic geometry \cite{dolan2009doctrines}.\newline

If one compares the development timelines of dimensional analysis and classical mechanics, the latter being a reference for the development of modern physical sciences at large, one can clearly observe a lapse of about 100 years between the two disciplines (see Figure \ref{timeline}).\newline

\begin{figure}[h]
\centering
\includegraphics[scale=0.3]{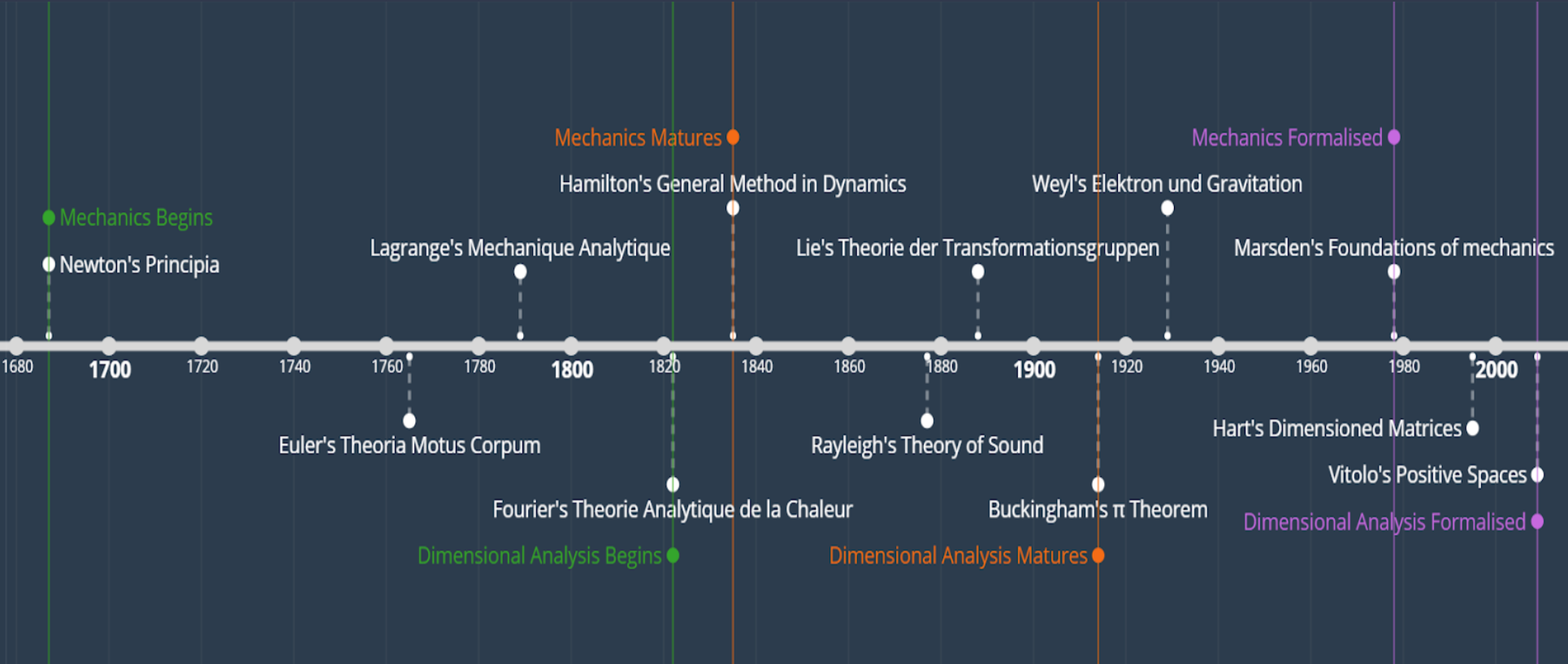}
\caption{A comparative timeline of the inception (green) and maturation (orange) of classical mechanics and dimensional analysis as measured by the bibliographical record. Mathematical formalisation (purple) is attributed to an axiomatization of the theory within modern mathematics. In the case of classical mechanics this corresponds to the development of geometric mechanics, which we link to the landmark works of Marsden and Arnold \cite{abraham1978foundations} \cite{arnold2013mathematical}.}
\label{timeline}
\end{figure}

This historical mismatch partially explains the lack of development in the mathematical foundations of metrology to this day, but there is another factor at play: the conventional system of extracting quantitative predictions from mathematical models, with all its folklore and ad hoc ways, just seems to work very well in practice. One could simply point to the International Bureau of Weights and Measures \cite{bimp2012metrology} as evidence for this fact. We are not here to reinvent the wheel; what we wish is to better understand the mathematical structure of physical quantities and how it fits with the usual formulations of modern physics.\newline

Our goal in the present work is not so much to define an algebraic theory of physical quantities -- this has been done quite successfully in \cite{janyvska2010algebraic} -- but to define a theory of algebra informed by how physical quantities are used in practice. To this end we shall motivate the definition of a generalised notion of set and binary operation and then proceed to develop an analogue of the ordinary theory of commutative algebra, discussing the generalizations of objects such as abelian groups, rings, modules and algebras.

\section{From Physical Quantities to Dimensioned Algebra} \label{dimmaths}

\subsection{The formal structure of dimensional analysis} \label{physicalquantity}

Consider the following high school science textbook exercise:\newline

``\textbf{Question}: A water dispenser has two taps, one dispenses $36.7\, \text{cm}^3/\text{s}$ and the other $2.1\, \text{L}/\text{min}$, how long will it take for an empty $300\, \text{cm}^3$ cup to fill up when placed under both taps? \textbf{Answer}: The flow through a pipe measures volume of fluid traversing the pipe per unit time, so combining both taps means that we should add the flows. To this end we write both flows in the same units using the appropriate conversion factors to find a combined flow of $2.2\, \text{L}/\text{min} + 2.1\, \text{L}/\text{min}=4.3\, \text{L}/\text{min}$. Setting $F=V/T$ with $F=4.3\, \text{L}/\text{min}$ and $V=300\, \text{cm}^3=0.3\, \text{L}$ we find $T=0.06\, \text{min}$ so the cup will fill up in approximately $4$ seconds, again using the appropriate conversion factors.''\newline

Observe how the relevant physical quantities, e.g. the volume of the cup, are specified by a numerical value ($300$) and a unit of measurement ($\text{cm}^3$) which is arbitrary within the class of units ($\text{cm}^3$, L, $\text{m}^3$...) measuring the same type of physical quantity (volume). We thus see that concrete physical quantities resulting from a practical measurement are equivalence classes of numerical values in all possible appropriate units.\newline

Another general observation is that algebraic operations correspond to specific physical phenomena, i.e. the addition of flows as physical quantities is correlated with the combination of water streams and the multiplication by rates is correlated with the passage of time. Although we will not delve further into this topic here, we note that this is rarely considered in the literature; most theoretical treatments simply postulate algebraic operations as formal devices without explicit connection to physical phenomena.\newline

The most relevant aspect to our discussion, however, is the peculiar way in which algebraic operations behave:
\begin{itemize}
    \item Addition can only be performed between quantities specified by the same unit of measurement and only affects the numerical part, addition is otherwise undefined.
    \item Multiplication can be performed between any two arbitrary physical quantities, it affects the numerical part and the units part.
    \item All algebraic operations are compatible with conversion factors that allow change between units of the same kind.
\end{itemize}
These are, of course, the usual \emph{rules of the game} of standard \textbf{dimensional analysis}. In general, one considers a family of $n$ physical quantity types, such as mass, time, volume, velocity, charge, etc., usually denoted by $[Q_1],\dots,[Q_n]$, and some system of units $u_1,\dots,u_n$, one for each quantity type. Two arbitrary physical quantities are commonly expressed as the following mathematical entities:
\begin{equation}
    Q=q\, u_1^{a_1}\cdots n_n^{a_n} \qquad P=p\, u_1^{b_1}\cdots n_n^{b_n}
\end{equation}
with numerical values $q,p\in \Real$ and exponents $a_i,b_i\in \Int$ (or perhaps $a_i,b_i\in\Rat$). Addition is only defined when $a_i=b_i$, in which case:
\begin{equation}
    Q+P=(q+p)\, u_1^{a_1}\cdots n_n^{a_n},
\end{equation}
and multiplication is generally defined via
\begin{equation}
    Q\cdot P = (qp)\, u_1^{a_1+b_1}\cdots n_n^{a_n+b_n}.
\end{equation}
From a mathematical point of view, the set of physical quantities is $\Real \times \Int^n$ with addition only partially defined in the first argument and multiplication defined as a direct product in both arguments. The $\Int^n$ component corresponds to the types of physical quantities and the domains of partial addition are precisely the subsets of matching type.

\subsection{Dimensioned Sets and Binars} \label{dimsets}

Following our discussion in Section \ref{physicalquantity}, we are led to consider `typed' or `labeled' sets as the primary objects of interest. Our approach is to formally define these `labelled sets' and to identify the natural categories that emerge from considering different compatibility conditions between the `labelling' structure and binary operations defined on the sets.\newline

The notion that physical quantities have `types' is captured simply by what we call a \textbf{dimensioned set} which is nothing but a surjection of sets $\delta:A\to D$. We call $D$ the \textbf{set of dimensions}, $\delta$ the \textbf{dimensionality projection} and the preimages $A_d:=\delta^{-1}(d)\subset A$ \textbf{dimension slices} of the set $A$. In the interest of brevity we may denote a dimensioned set $\delta:A\to D$ simply as $A_D$. A \textbf{morphism of dimensioned sets} or \textbf{dimensioned map} $\Phi_\varphi:A_D\to B_E$ is simply a morphism of surjections, that is, a commutative diagram:
\begin{equation}
\begin{tikzcd}
A \arrow[r, "\Phi"] \arrow[d, "\delta"'] & B \arrow[d, "\epsilon"] \\
D \arrow[r, "\varphi"'] & E
\end{tikzcd}
\end{equation}
The \textbf{category of dimensioned sets} is denoted by $\textsf{DimSet}$. The \textbf{cartesian product} of two dimensioned sets $A_D \times B_E$ is defined in the obvious way:
\begin{equation}
\begin{tikzcd}
A \times B \arrow[d, "\delta \times \epsilon"]\\
D \times E
\end{tikzcd}
\end{equation}
A distinguished singleton considered as a dimensioned set $\{\bullet\}\to \{\bullet\}$ clearly acts as a terminal object in $\textsf{DimSet}$ and as a unit for $\times$. The unit object is denoted by $1$ and will be called the \textbf{dimensionless set}. It is then easy to see that $(\textsf{DimSet},\times, 1)$ forms a cartesian monoidal category. Further assuming that $1$ is an initial object in $\textsf{DimSet}$ corresponds to assuming that all dimension sets have a distinguished element representing the `dimensionless' dimension; we shall see in later sections that it is sometimes useful to make such an assumption.\newline

It will be useful to introduce a notation that reflects the dimensioned structure explicitly so that we can keep track of the consistency of expressions. In what follows, unless redundant, elements of a dimensioned set will be denoted with a subscript indicating its dimension projection: for a dimensioned set $\delta:A\to D$ we denote
\begin{equation}
    a_d\in A \, \text{ where } \, d=\delta(a)\in D.
\end{equation}
With this convention, the action of a dimensioned map $\Phi_\varphi$ is notated $\Phi(a_d)=\Phi(a)_{\varphi (d)}$.\newline

Let us now discuss binary operations on dimensioned sets. For the remainder of this section it will be useful to think of dimensioned sets as disjoint unions of their dimension slices:
\begin{equation}
    A_D=\bigcup_{d\in D} A_d.
\end{equation}
General, possibly partially-defined, binary operations on dimensioned sets appear, in principle, considerably more nuanced that ordinary binary operations. Natural compatibility conditions between binary operations and the dimensioned structure may be imposed in multiple ways. Without any further choices, however, there are two canonical types of binar-like structure whose composition domains intersect with the dimension slices in extremal ways: on the one hand binar, we have structures whose composition domains are the dimension slices, and which will be called dimensional binars; and on the other, binar structures whose composition domain is the entire set together with a transitivity condition between dimension slices, and which will be called dimensioned binars. Intuitively, dimensional binars are operations \emph{strictly within} the dimension slices and dimensioned binars are operations \emph{strictly between} the dimension slices.\newline

Let a dimensioned set $\delta:A\to D$, then a \textbf{dimensional binar} structure $(A_D,*_D)$ is a partially-defined binary operation on $A$ with
\begin{equation}
    a *_d b \, \text{ defined only when } \, \delta(a)=\delta(a * b)=\delta(b)=d.
\end{equation}
This means that in the expression $a_d*_db_d$ all subscripts must agree for the product to be defined. In other words, a dimensional binar is a collection of ordinary binars indexed by the set of dimensions $\{(A_d,*_d), \, d\in D\}$. A \textbf{morphism of dimensional binars} $\Phi_\varphi:(A_D,*_D)\to (B_E,\circ_E)$ is a dimensioned map $\Phi_\varphi$ such that
\begin{equation}
    \forall \, d\in D, \, a,b\in A_d\qquad \Phi(a*_d b)=\Phi(a)\circ_{\varphi(d)} \Phi(b).
\end{equation}
Again, a morphism of dimensional binars can be regarded as a collection of ordinary morphisms of binars between the dimension slices:
\begin{equation}
    \{\Phi_d:(A_D,*_d)\to (B_{\varphi(d)},\circ_{\varphi(d)}), \, d\in D\}.
\end{equation}
The \textbf{category of dimensional binars} is denoted by $\textsf{DimBin}$. The cartesian product $\times$ on $\textsf{DimSet}$ extends to the category dimensional binars via the obvious construction: given two dimensional binars $(A_D,*_D)$ and $(B_E,\circ_E)$ there is a dimensional binar structure on the product $(A_D\times B_E, \bullet_{D\times E})$ where:
\begin{equation}
    (a_d,b_e)\bullet_{(d,e)} (a_d',b_e'):=(a_d*_da_d',b_e\circ_e b_e')
\end{equation}
for all $a_d\in A_D$ and $b_e\in B_E$.\newline

Interestingly, in direct analogy with ordinary sets and binars, the sets of morphisms of dimensional binars carry natural dimensional binar structure. A trivial observation is that morphisms of dimensioned sets, which are simply morphisms of surjections, have a natural surjection into maps of sets, i.e. $\pi: \Phi_\varphi \mapsto \varphi$. This means that for any two dimensioned sets $A_D$ and $B_E$ the set of morphisms together with the natural surjection $\pi: \textsf{DimSet}(A_D,B_E)\to \textsf{Set}(D,E)$ is a dimensioned set. Furthermore, when there are dimensional binar structures on the sets $(A_D,*_D)$ and $(B_E,\circ_E)$, a point-wise construction endows $\textsf{DimSet}(A_D,B_E)$ with a dimensional binar structure: given two morphisms $\Phi_\varphi, \Psi_\psi\in \textsf{DimSet}(A_D,B_E)$, for all $a_d\in A_d$ define
\begin{equation}
    \Phi_\varphi \circ \Psi_\psi (a_d):=\Phi(a)_{\varphi(d)} \circ_e \Psi(a)_{\psi(d)}
\end{equation}
which is indeed only possible when $\varphi(d)=e=\psi(d)$. Hence the point-wise operation so defined will endow the set of dimensioned maps with a dimensional binar structure:
\begin{equation}
    (\textsf{DimSet}(A_D,B_E)_{\textsf{Set}(D,E)},\,\circ_{\textsf{Set}(D,E)}).
\end{equation}

Let a dimensioned set $\delta:A\to D$, a \textbf{dimensioned binar} structure $(A_D,*^D)$ is a totally-defined binary operation $*$ on $A$ with
\begin{equation}
    \forall \, d,d'\in D, \, \exists!\, d''\in D \, : \qquad A_d* A_{d'}\subset A_{d''}.
\end{equation}
This condition is equivalent to the set of dimensions carrying a binar structure $(D,\,)$ (denoted by juxtaposition) and the dimension projection being a morphism of binars $\delta: (A,*)\to (D,\,)$. We thus write $a_d*b_e=(a*b)_{de}$, where the juxtaposition $de$ denotes de binar structure of the set of dimensions. A \textbf{morphism of dimensioned binars} between $(A_D,*^D)$ and $(B_E,\circ^E)$ is simply a dimensioned map $\Phi_\varphi:A_D\to B_E$ such that $\Phi:(A,*)\to (B,\circ)$ is a morphism of binars, since the binary operations are totally-defined. Note that this condition on $\Phi_\varphi$ makes $\varphi: (D,\,) \to (E,\,)$ into a morphism of binars.\newline

Let us now consider interactions between two (or more) binary operations on dimensioned sets. In the most general scenario, a set $A$ has two independent dimensioned structures $\delta: A\to D$ and $\epsilon: A\to E$ and two binary operations, one relative to the $D$ dimensions and the other to the $E$ dimensions. It is easy to see that a systematic analysis of compatibility conditions between two such binary operations becomes almost impossible to do in full generality due to the partial nature of the operations. The situation becomes much more tractable when only a single dimensioned structure is assumed on the set $A_D$ and the binary operations are all defined relative to it. This case, which we shall take for the remainder of this work, is further justified by the motivating example of physical quantities, where there is only one general notion of \emph{physical dimension} for a collection of measurable physical quantities.\newline

Focusing now on the case of interest of a dimensioned set $A_D$ with two binar structures, we find three possible cases depending on whether the operations are dimensional or dimensioned:
\begin{equation}
    \text{i)}\, (A_D,*_D,\circ_D) \qquad \text{ii)}\, (A_D,*^D,\circ^D) \qquad \text{iii)}\, (A_D,*_D,\circ^D).
\end{equation}
Cases i) and ii) in fact reduce to the ordinary theory of binars: collections of pairs of binar structures indexed by $D$ in the case of i) and two binar structures on $D$ in the case of ii). Case iii) seems to point at genuinely new possibilities for the interaction of two binar structures. This will indeed be confirmed by the results in the sections to follow.

\section{Dimensioned Rings} \label{dimring}

Our goal is to replicate the standard theory of algebraic structures, i.e. groups, rings, modules, etc., while attempting to account for the characteristic structure of physical quantities discussed in Section \ref{physicalquantity}. We shall see that developing such a theory is natural and straightforward when working with dimensioned sets and dimensioned binary operations as introduced in Section \ref{dimsets}.\newline

The theory of dimensional binars can be extended trivially to include familiar notions such as identity elements, associativity or invertibility. Consider a dimensional binar $(A_D,*_D)$ and some abstract property of binars $P$ (e.g. commutativity, existence of identity, etc.), we will say that $(A_D,*_D)$ \textbf{satisfies property} $P$ simply when the binar slices $(A_d,*_d)$ satisfy the property $P$ for all $d\in D$.\newline

Since we are aiming to identify ring-like structures where addition is partially-defined we begin by defining dimensional abelian groups. We say that $(A_D.+_D)$ is a \textbf{dimensional abelian group} when $(A_d,+_d)$ is an abelian group for all $d\in D$. Morphisms of dimensional abelian groups are simply morphisms of underlying dimensional binar structures. It follows from the general results for dimensional binars of Section \ref{dimsets} that the set of morphisms between two dimensional abelian groups $(A_D,+_D)$ and $(B_E,+_E)$ has the structure of a dimensional abelian group with dimension set given by the set of maps between the dimension sets $D$ and $E$. This will be called the set of \textbf{dimensioned morphisms} or \textbf{dimensioned maps} and we will denote it by $(\Dim{A_D,B_E}_{\text{Map}(D,E)},+_{\text{Map}(D,E)})$ or for the endomorphisms of a single dimensioned group $\Dim{A_D}_{\text{Map}(D)}:=\Dim{A_D,A_D}_{\text{Map}(D,D)}$. Subscripts will be omitted whenever they can be inferred from context. The \textbf{category of dimensional abelian groups} will be denoted by $\textsf{DimAb}$.\newline

Dimensional abelian groups display structures analogous to those of ordinary abelian groups. Firstly, subgroups, products and quotients can be naturally generalised to dimensional abelian groups. Let $(A_D,+_D)$ be a dimensional abelian group, then the subset $0_D:=\{0_d\in (A_d,+_d),\quad d\in D\}$ is called the \textbf{zero} of $A_D$. A subset $S\subset A_D$ is called a \textbf{dimensional subgroup} when $S_d:=S\cap A_d\subset (A_d,+_d)$ are subgroups for all $d\in D$. A dimensional subgroup $S\subset A_D$ is clearly a dimensional group with dimension set given by $\delta(S)$, where $\delta: A\to D$ is the dimension projection. We can define the \textbf{kernel} of a dimensional group morphism $\Phi:A_D\to B_E$ in the obvious way
\begin{equation}
    \Ker{\Phi}:=\{a_d\in A_D|\quad \Phi(a_d)=0_{\phi(d)}\}.
\end{equation}
Clearly, the zero $0_D$ and kernels of dimensional morphisms $\Ker{\Phi}\subset A_D$ are examples of dimensional subgroups. A dimensional subgroup $S\subset A_D$ whose slice intersections $S\cap A_d\subset (A_d,+_d)$ are normal subgroups also induces a natural notion of \textbf{quotient}:
\begin{equation}
    A_D/S:=\bigcup_{d\in \delta(S)} A_d/(S\cap A_d)
\end{equation}
which has an obvious dimensioned group structure with dimension set $\delta(S)$. There is also a natural notion of \textbf{product} of two dimensional groups $A_D$, $B_E$ given by the categorical product of dimensional binars $(A_D\times B_E, +_{D\times E})$. Furthermore, when we fix a dimension set $D$ and we consider \textbf{dimension-preserving morphisms}, i.e. dimensional group morphisms $\Phi:A_D\to B_D$ for which the induced map on the dimension sets is the identity $\Id_D:D\to D$, the dimensional abelian dimensional groups over $D$ form a subcategory $\textsf{DimAb}_D\subset\textsf{DimAb}$ that, in addition to the notions of subgroup, kernel and quotient, also admits a \textbf{direct sum} defined as $A_D\oplus_D B_D:=(A\times B)_D$ with partial multiplication given in the obvious way
\begin{equation}
    (a_d,b_d)+_d (a_d',b_d'):=(a_d+_d a_d',b_d +_d b_d').
\end{equation}
It is easy to prove that this direct sum operation on $\textsf{DimAb}_D$ acts as a product and coproduct for which the notions of kernel and quotient identified in the general category $\textsf{DimAb}$ satisfy the axioms of an abelian category. We call $\textsf{DimAb}_D$ the \textbf{category of $D$-dimensional abelian groups}. These  constructions are indeed identical to those commonly defined within the categories of abelian group bundles.\newline

The dimension structure of a dimensioned set $\delta: S\to D$ allows for an obvious generalization of the free abelian group construction: the \textbf{dimensional free abelian group} on a dimensioned set $S_D$ is defined as the dimensional group
\begin{equation}
    \Int[S_D]:=\bigcup_{d\in D} \Int[A_d]
\end{equation}
where $\Int[\,\,]$ denotes the usual free abelian group construction. The dimensioned version of the free abelian group can be applied to the cartesian product of dimensional abelian groups to obtain the natural generalization of the tensor product for dimensional abelian groups: let two dimensional abelian groups $(A_D,+_D)$ and $(B_E,+_D)$, their \textbf{tensor product} is defined simply as follows:
\begin{equation}
    A_D\otimes B_E := \bigcup_{(d,e)\in D\times E} A_d\otimes B_e
\end{equation}
where $\otimes$ denotes the ordinary tensor product of abelian groups.\newline

Having identified the structure of a partially-defined additive operation as a dimensional abelian group, we are now in the position to attempt a definition of dimensioned ring by considering a multiplicative operation together with the additive operation. In accord with the guiding example of physical quantities, let us consider a dimensional abelian group $(A_D,+_D)$ with a total multiplication $\cdot$, i.e. assume $(A, \cdot)$ is a monoid. The key axiom that characterises rings across multiple conventions is distributivity of multiplication with addition, however the partial nature of addition may present an obstruction to demanding distributivity in general. If one considers three elements $a_d,b_e,c_f\in A_D$ and attempts to write the (right) distributivity property:
\begin{equation}
    (a_d+b_e)\cdot c_f = a_d\cdot c_f + b_e\cdot c_f,
\end{equation}
it is clear that a few compatibility conditions in dimensions are required for distributivity to possibly hold in some generality within $A$. The fact that addition can only happen between elements of the same dimension means that in the above formula $d=e$ and that the dimension of $a_d\cdot c_f$ and $b_d\cdot c_f$ must be the same. Therefore, if we are to demand distributivity as generally as possible, the multiplicative operation must map transitively between dimension slices, in other words, the dimension of $a_d\cdot c_f$ only depends on $d$ and $f$. This is precisely the notion of a dimensioned binary operation introduced in Section \ref{dimsets} and thus we are left with a clear motivation to consider case iii) of the possible ways in which two binary operations interact on a dimensioned set.\newline

A \textbf{dimensioned ring} is a triple $(R_D,+_D,\cdot^D)$ where $R_D$ is a dimensioned set, $(R_D,+_D)$ is a dimensional abelian group, $(R_D,\cdot^D)$ is a dimensioned monoid and the distributivity condition
\begin{equation}
    (a+b)\cdot c = a\cdot c + b\cdot c \qquad c\cdot (a+b) = c\cdot a + c\cdot b
\end{equation}
holds whenever it is defined for $a,b,c\in R$. Recall from the definition of dimensioned binar in Section \ref{dimsets} that the dimension projection $\delta: R\to D$ becomes a morphism of binars and so the dimension set $D$ of a dimensioned ring $(R_D,+_D,\cdot^D)$ carries an associative unital binary operation $(D,\, )$, denoted by juxtaposition, such that $\delta: (R,\cdot)\to (D,\, )$ is a monoid morphism. With a slight abuse of notation we denote by $1$ the multiplicative identities of both $(R,\cdot)$ and $(D,\, )$, thus symbolically $1=\delta(1)$. The slice containing $1\in R$, or, equivalently, above $1\in D$ is called the \textbf{dimensionless slice} of the dimensioned ring $R_1\subset R$. A dimensioned ring is called \textbf{commutative} when the monoid structures are abelian. For the remainder of this text dimensioned rings and ordinary rings are assumed to be commutative unless otherwise stated.\newline

Let $(R_D,+_D,\cdot^D)$ and $(P_E,+_E,\cdot^E)$ be two dimensioned rings, a dimensioned map $\Phi:R_D\to P_E$ is called a \textbf{morphism of dimensioned rings} when
\begin{equation}
    \Phi(a\cdot b)=\Phi(a)\cdot \Phi(b), \qquad \Phi(1_R)=1_P
\end{equation}
for all $a,b\in R_D$. The map between the dimension monoids $\phi:D\to E$ is thus necessarily a monoid morphism. Dimensioned rings with these morphisms form the \textbf{category of dimensioned rings}, denoted by \textsf{DimRing}.\newline

The \textbf{product} of two dimensioned rings $R_D\times P_E$ is defined as the obvious extension of the product of dimensional abelian groups above and the product of ordinary monoids. A dimensioned subgroup $S\subset (R_D,+_D)$ is called a \textbf{dimensioned subring} when $S\cdot S\subset S$ and $1\in S$. A dimensioned subring $I\subset R_D$ is called a \textbf{dimensioned ideal} if for all elements $a_d\in R_D$ and $i_e\in I$ we have
\begin{equation}
    a_d\cdot i_e\in I\cap R_{de}.
\end{equation}
Note that the zero $0_D\subset R$ is an ideal since the dimensioned ring axioms imply that it acts as an absorbent set in the following sense:
\begin{equation}
    0_d\cdot a_e=0_{de}.
\end{equation}

\begin{prop}[Quotient Dimensioned Ring] \label{QuotDimRing}
Let $(R_D,+_D,\cdot^D)$ be a dimensioned ring and and $I\subset R_D$ a dimensioned ideal, then the quotient dimensioned group $R/I$ carries a canonical dimensioned ring structure such that the projection map:
\begin{equation}
    q: R\to R/I
\end{equation}
is a morphism of dimensioned rings. This construction is called the \textbf{quotient dimensioned ring}.
\end{prop}
\begin{proof}
Let us denote dimension slices of the ideal by $I_d:= I\cap R_d$. From the construction of quotient dimensioned group we see that the projection map $q: R\to R/I$ is explicitly given by
\begin{equation}
    a_d\mapsto a_d + I_d,
\end{equation}
which makes $R/I$ into a dimensioned abelian group with dimension set $\delta (I)\subset D$. The dimensioned ring multiplication on the quotient can be explicitly defined by:
\begin{equation}
    (a_d+_dI_d)\cdot (b_e+_eI_e)=a_d\cdot b_e +_{de} a_d\cdot I_e +_{de} b_e\cdot I_d +_{de} I_d\cdot I_e=a_d\cdot b_e +_{de} I_{de}.
\end{equation}
This is easily checked to be well-defined and to inherit all the dimensioned ring multiplication properties from $R_D$. The map $q$ is then a morphism of dimensioned rings by construction. Note that the quotient ring has dimension projection $\delta': R/I \to \delta(I)$ thus, in particular, $\delta(I)\subset D$ is a submonoid.
\end{proof}

A \textbf{unit} or \textbf{choice of units} $u$ in a dimensioned ring $R_D$ is a section of the dimension projection
\begin{equation}
\begin{tikzcd}[row sep=small]
R  \arrow[d, "\delta"'] \\
D \arrow[u, "u"', bend right=60] 
\end{tikzcd}\quad \delta \circ u =\Id_D,
\qquad \text{ such that } \quad
    u_{de}=u_d\cdot u_e \quad \text{and} \quad u_d\neq 0_d
\end{equation}
for all $d,e\in D$. In other words, a unit is a splitting $u:(D,\,)\to (R,\cdot)$ of the monoid surjection $\delta: (R,\cdot)\to (D,\,)$ with non-zero image. Units can be regarded as the dimensioned generalization of the notion of non-zero element of a ring with the caveat that they may not exist due to the non-vanishing condition being required for all of $D$. It was noted above that vector bundles are a form of extreme example of dimensioned rings; in this vein, considering the Moebius band as a dimensioned ring with dimension set the circle $\text{S}^1$ and the zero multiplication operation, we find an explicit example of a dimensioned ring that does not admit units, since they would correspond to global non-vanishing sections of a non-trivialisable vector bundle.\newline

It turns out that all dimensioned rings carry ordinary ring structures within them.

\begin{prop}[Dimensionless Ring] \label{DimLessRing}
Let $(R_D,+_D,\cdot^D)$ be a dimensioned ring, then its dimensionless slice $R_1$ carries a natural ring structure induced from the partial addition defined on the slice $+_1$ and the restriction of the total multiplication $\cdot|_{R_1}$. This is called the \textbf{dimensionless ring} $(R_1,+_1,\cdot|_{R_1})$. Furthermore, morphisms, products and quotients of dimensioned rings induce their analogous counterparts for dimensionless rings.
\end{prop}
\begin{proof}
The fact that $(R_1,+_1,\cdot|_{R_1})$ is a ring stems simply from the fact that $R_1$ is closed under multiplication, since $1\in D$ is the monoid identity. From a similar reasoning we see that for any dimensioned ideal $I\subset R_D$ the dimensionless slice $I_1:= I\cap R_1$ is an ideal of the dimensionless ring. Let two dimensioned rings $(R_D,+_D,\cdot^D)$ and $(P_E,+_E,\cdot^E)$. Since a morphism $\Phi: R_D \to P_E$ preserves the multiplicative identities, it is clear that the restriction
\begin{equation}
    \Phi|_{R_1}:(R_1,+_1,\cdot|_{R_1})\to (P_{\phi(1)},+_{\phi(1)},\cdot|_{P_{\phi(1)}})
\end{equation}
is a morphism of rings. The dimensionless slice of the dimensioned product $R_D \times P_E$ is indeed $R_1\times P_1$ with the direct product construction applying to rings in a straightforward manner.
\end{proof}
This shows that dimensioned rings are, in fact, a strict generalization of ordinary rings since we recover them by considering trivial dimension monoids, i.e. singleton dimension sets. More precisely, the category of rings is a subcategory of the category of dimensioned rings $\textsf{Ring}\subset \textsf{DimRing}$.\newline

Thus far we have seen that abelian group bundles and ordinary rings are extreme examples of dimensioned rings: the former by taking the zero multiplication and the latter by taking the dimension set to be a singleton. A natural example of dimensioned ring that is somewhat intermediate to the aforementioned two is what we call a \textbf{product dimensioned ring}: let $(R,+,\cdot)$ be a ring and $(D,\,)$ a monoid, then the cartesian product $R\times D$ carries a natural dimensioned ring structure defined in the obvious way
\begin{equation}
    \Proj_2:R\times D\to D \qquad (a,d)+_d(b,d):=(a+b,d), \qquad (a,d)\cdot (b,e):=(a\cdot b,de).
\end{equation}
The dimensionless ring of a product dimensioned ring is simply $(R\times D)_1=R\times \{1\}$ and a unit in $R\times D$ is given by a monoid morphism $u:(D,\,) \to (R,\cdot)$ such that $u(d)\neq 0$ for all $d\in D$. Note that a product dimensioned ring always admits a unit given by the constant map $1: D \to R$ such that $1(d)=1$ for all $d\in D$.\newline

The first non-trivial example of a dimensioned ring is the set of dimensioned maps of a dimensional abelian group.

\begin{prop}[Endomorphism Ring] \label{DimEndRing}
Let $(A_D,+_D)$ be a dimensional abelian group, then the set of dimensioned maps $\normalfont\Dim{A_D}_{\text{Map}(D)}$ carries a dimensioned (non-commutative) ring structure
\begin{equation}
\normalfont
    (\Dim{A_D}_{\text{Map}(D)},+_{\text{Map}(D)},\circ^{\text{Map}(D)})
\end{equation}
where $+$ denotes the dimensional abelian addition and $\circ$ is composition.
\end{prop}
\begin{proof}
The set of dimensioned maps $\Dim{A_D}_{\text{Map}(D)}$ has a canonical dimensional abelian group structure as shown for the category of dimensional binars in Section \ref{dimsets}. Note that the dimension set, the maps from $D$ into itself $\text{Map}(D)$, carries a natural monoid structure given by composition of maps $\circ$. Let the dimensioned maps $\Phi_\phi,\Theta_\phi, \Psi_\psi:A_D\to A_D$ where $\phi,\psi:D\to D$ are the dimension maps. It follows by construction that
\begin{equation}
    \Phi_\phi \circ \Psi_\psi = (\Phi \circ \Psi)_{\phi\circ \psi},
\end{equation}
so we see that composition is indeed a dimensioned binary operation $\circ^{\text{Map}(D)}$. It only remains to check the distributivity property:
\begin{equation}
    (\Phi_\phi +_\phi \Theta_\phi )\circ \Psi_\psi= \Phi_\phi\circ \Psi_\psi +_{\phi\circ \psi} \Theta_\phi \circ \Psi_\psi \qquad \Psi_\psi \circ (\Phi_\phi +_\phi \Theta_\phi ) = \Psi_\psi \circ \Phi_\phi +_{\psi \circ \phi} \Psi_\psi \circ \Theta_\phi
\end{equation}
which follows from the point-wise definition of addition and the fact that the maps are dimensional abelian group morphisms.
\end{proof}

Division in dimensioned rings is formally analogous to division in ordinary rings since the multiplication operation is totally defined. An element of a dimensioned ring $a\in R_D$ is said to be \textbf{invertible} if there exists a (necessarily unique) element $1/a\in R_D$, called its \textbf{reciprocal}, such that $a\cdot 1/a = 1$. Subtleties of the dimensioned case appear, however, when considering the notion of \textbf{zero divisor} as an element $z\in R_D$ such that there exists a $z'\in R_D$ with $z\cdot z' \in 0_D$. We will not delve further into these questions here.\newline

A dimensioned ring $R_D$ is called a \textbf{dimensioned field} when all non-zero elements are invertible. Note that for this requirement to be consistent with the dimension projection $\delta: R\to D$, the monoid structure on $D$ must be a group. A direct consequence of the defining condition of dimensioned field is that non-zero elements induce bijective maps between dimension slices. Indeed, for a non-zero element $0_d\neq a_d\in R_D$ we have the following induced maps called \textbf{slice-wise multiplications}:
\begin{align}
a_d\cdot : R_e & \to R_{de}\\
b_e & \mapsto a_d\cdot b_e
\end{align}
for all $e\in D$. The distributivity axiom implies that these are slice-wise abelian group isomorphisms with inverse given by $1/(a_d)\cdot$. These maps allow to prove a general result that confers a role to choices of unit on dimensioned fields similar to that of a trivialization of a fibre bundle.

\begin{prop}[Units in Dimensioned Fields] \label{UnitsDimFields}
Let $(R_D,+_D,\cdot)$ be a dimensioned field, then a choice of units $u:D\to R$ induces an isomorphism with the product dimensioned field:
\begin{equation}
    R_D\cong R_1\times D.
\end{equation}
\end{prop}
\begin{proof}
A choice of units induces the following map via slice-wise multiplication:
\begin{align}
\Phi_u: R_1\times D & \to R_D\\
(r,d) & \mapsto u_d\cdot r
\end{align}
This is shown to be a bijection by explicitly constructing its inverse $\Phi_u^{-1}(a_d):=u_{d^{-1}}\cdot a_d$. It only remains to check that $\Phi_u$ is dimensioned ring morphism; this follows directly by construction and the fact that $u$ is a morphism of monoids:
\begin{align}
    \Phi^u((r_1,d)\cdot (r_2,e))&= \Phi^u((r_1\cdot r_2,de))=u_{de}\cdot r_1\cdot r_2=u_d\cdot u_e\cdot r_1\cdot r_2 =\\
    &=(u_d\cdot r_1)\cdot (u_e\cdot r_2)=\Phi^u(u_d\cdot r_1)\cdot \Phi^u(u_e\cdot r_2).
\end{align}
\end{proof}
This last proposition shows that the dimensioned fields for which choices of units exist are (non-canonically) isomorphic to the product dimensioned fields $F\times D$ with $F$ an ordinary field and $D$ an abelian group.

\section{The Power Functor} \label{power}

In this section we describe an important class of examples of dimensioned rings arising from ordinary 1-dimensional vector spaces. These examples are important for two reasons: on the one hand, mathematically, they constitute a large class of natural non-trivial examples of dimensioned rings, on the other they capture the standard structure of physical quantities described in Section \ref{physicalquantity} precisely.\newline

We identify \textbf{the category of lines}, $\Line$, as a subcategory of vector spaces over a field $\Vect_{\mathbb{F}}$. Objects are vector spaces of dimension 1, a useful way to think of these in the context of the present work is as sets of numbers without the choice of a unit. An object $L\in\Line$ will be appropriately called a \textbf{line}. A morphism in this category $b\in\text{Hom}_{\Line}(L,L')$, usually simply denoted by $b:L\to L'$, is an invertible (equivalently non-zero) linear map. Composition in the category $\Line$ is simply the composition of maps. If we think of $L$ and $L'$ as numbers without a choice of a unit, a morphism $b$ between them can be thought of as a unit-free conversion factor, for this reason we will often refer to a morphism of lines as a \textbf{factor}. We consider the field $\Field$, trivially a line when regarded as a vector space, as a singled out object in the category of lines $\Field\in \Line$.\newline

It is a simple linear algebra fact that any two lines $L,L'\in\Line$ satisfy
\begin{equation}
    \dimm (L \oplus L')= \dimm L + \dimm L' = 2 > 1, \qquad \dimm L^* = \dimm L = 1, \qquad \dimm (L \otimes L') = 1.
\end{equation}
Then, we note that the direct sum $\oplus$, is no longer defined in $\Line$, however, it is straightforward to check that $(\Line, \otimes, \mathbb{F})$ forms a symmetric monoidal category and that $*:\Line \to \Line$ is a duality contravariant autofunctor. Let us introduce the following notation:
\begin{equation}
    \begin{cases} 
      L^n:=\otimes^nL & n>0 \\
      L^n:= \Field & n=0 \\
      L^n:=\otimes^nL^* & n<0 
   \end{cases}
\end{equation}
which is such that given two integers $n,m\in\Int$ and any line $L\in\Line$ the following equations hold
\begin{equation}
    (L^n)^* = L^{-n} \qquad L^{n}\otimes L^{m} = L^{n+m}.
\end{equation}
Thus we see how one single line and its dual $L,L^*\in\Line$ generate an abelian group with the tensor product as group multiplication, the patron $\Field\in\Line$ as group identity and the duality autofunctor as inversion. We define the \textbf{power} of a line $L\in\Line$ as the set of all tensor powers
\begin{equation}
    L^\odot:=\bigcup_{n\in \Int} L^n.
\end{equation}
This set has than an obvious dimensioned set structure with dimension set $\Int$:
\begin{equation}
    \pi:L^\odot\to \Int.
\end{equation}
Since dimension slices are precisely the tensor powers $L^n$, they carry a natural $\Field$-vector space structure, thus making the power of $L$ into a dimensional abelian group $(L^\odot_\Int,+_\Int)$. The next proposition shows that the ordinary $\Field$-tensor product of vector spaces endows $L^\odot$ with a dimensioned field structure.

\begin{prop}[Dimensioned Ring Structure of the Power of a Line] \label{DimRingPower}
Let $L\in\Line$ be a line and $(L^\odot_\Int,+_\Int)$ its power, then the $\Field$-tensor product of elements induces a dimensioned multiplication
\begin{equation}
    \odot: L^\odot \times L^\odot\to L^\odot
\end{equation}
such that $(L^\odot_\Int,+_\Int,\odot)$ becomes a dimensioned field.
\end{prop}
\begin{proof}
The construction of the dimensioned ring multiplication $\odot$ is done simply via the ordinary tensor product of ordinary vectors and taking advantage of the particular properties of 1-dimensional vector spaces. The two main facts that follow from the 1-dimensional nature of lines are: firstly, that linear endomorphisms are simply multiplications by field elements
\begin{equation}
    \text{End}(L)\cong L^*\otimes L\cong \Field
\end{equation}
which, at the level of elements, means that
\begin{equation}
    \text{End}(L)\ni\alpha \otimes a =\alpha(a)\cdot \Id_{L}
\end{equation}
as it can be easily shown by choosing a basis; and secondly, that the tensor product becomes canonically commutative, since, using the isomorphism above, we can directly check
\begin{equation}
    a\otimes b(\alpha,\beta)=\alpha(a)\beta(b)=\alpha(b)\beta(a)=b\otimes a(\alpha,\beta),
\end{equation}
thus showing
\begin{equation}
    a\otimes b=b\otimes a \in L\otimes L=L^2.
\end{equation}
The binary operation $\odot$ is then explicitly defined for elements $a,b\in L=L^1$, $\alpha,\beta\in L^*=L^{-1}$ and $r,s\in \Field=L^0$ by
\begin{align}
    a\odot b &:= a\otimes b\\
    \alpha\odot \beta &:= \beta \otimes \alpha\\
    r\odot s &:= r\otimes s=rs\\
    r\odot a &:= ra\\
    r\odot \alpha &:= r\alpha\\
    \alpha \odot a &:=\alpha(a) = a(\alpha) =: a \odot \alpha
\end{align}
Products of two positive power tensors $a_1\otimes \cdots \otimes a_q$, $b_1\otimes \cdots \otimes b_p$ and negative powers $\alpha_1\otimes \cdots \otimes \alpha_q$, $\beta_1\otimes \cdots \otimes \beta_p$ are defined by
\begin{align}
    (a_1\otimes \cdots \otimes a_q)\odot(b_1\otimes \cdots \otimes b_p) &:= a_1\otimes \cdots \otimes a_q\otimes b_1\otimes \cdots \otimes b_p\\
    (\alpha_1\otimes \cdots \otimes \alpha_q)\odot (\beta_1\otimes \cdots \otimes \beta_p) &:= \alpha_1\otimes \cdots \otimes \alpha_n \otimes \beta_1\otimes \cdots \otimes \beta_m
\end{align}
and extending by $\Field$-linearity. Let $q,p>0$, the dimensioned ring product satisfies:
\begin{equation}
    \odot: L^q \times L^p \to L^{q+p}, \qquad \odot: L^{-q} \times L^{-p} \to L^{-q-p}, \qquad \odot: L^0 \times L^0 \to L^0.
\end{equation}
For products combining positive power tensors $a_1\otimes \cdots \otimes a_q$ and negative power tensors $\alpha_1\otimes \cdots \otimes \alpha_p$ we critically make use of the isomorphism $L^*\otimes L\cong \Field$ to define without loss of generality:
\begin{equation}
    (a_1\otimes \cdots \otimes a_q) \odot (\alpha_1\otimes \cdots \otimes \alpha_p) := \alpha_1( a_1) \cdots \alpha_q( a_q) \alpha_{p-q}\otimes \cdots \otimes \alpha_p.
\end{equation}
It is then clear that the multiplication $\odot$ satisfies, for all $m,n\in \Int$,
\begin{equation}
    \odot: L^m \times L^n \to L^{m+n}
\end{equation}
and so it is compatible with the dimensioned structure of $L^\odot_\Int$. The multiplication $\odot$ is clearly associative and bilinear with respect to addition on each dimension slice from the fact that the ordinary tensor product is associative and $\Field$-bilinear. Then it follows that $(L^\odot_\Int,+_\Int,\odot)$ is a commutative dimensioned ring. It only remains to show that non-zero elements of $L^\odot$ have multiplicative inverses. Note that a non-zero element corresponds to some non-vanishing tensor $0\neq h\in L^n$, but, since $L^n$ is a 1-dimensional vector space for all $n\in \Int$, we can find a unique $\eta\in (L^n)^*=L^{-n}$ such that $\eta(h)=1$. It follows from the above formula for products of positive and negative tensor powers that, in terms of the dimensioned ring multiplication, this becomes
\begin{equation}
    h\odot \eta =1,
\end{equation}
thus showing that all non-zero elements have multiplicative inverses, making the dimensioned ring $(L^\odot_\Int,+_\Int,\odot)$ into a dimensioned field.
\end{proof}

The construction of the power dimensioned field of a line is, in fact, functorial.

\begin{thm}[The Power Functor for Lines] \label{PowerIsAFunctorLine}
The assignment of the power construction to a line is a functor
\begin{equation}
\normalfont
    \odot: \Line \to \textsf{DimRing}.
\end{equation}
Furthermore, a choice of unit in a line $L\in \Line$ induces a choice of units in the dimensioned field $(L^\odot_\Int,+_\Int,\odot)$ which, since $L^0=\Field$, then gives an isomorphism with the product dimensioned field
\begin{equation}
    L^\odot \cong \Field \times \Int.
\end{equation}
\end{thm}
\begin{proof}
To show functoriality we need to define the power of a factor of lines $B:L_1\to L_2$
\begin{equation}
    B^\odot:L_1^\odot \to L_2^\odot.
\end{equation}
This can be done explicitly in the obvious way, for $q>0$
\begin{align}
    B^\odot|_{L^q} &:= B\otimes \stackrel{q}{\cdots} \otimes B :L_1^q\to L_2^q\\
    B^\odot|_{L^0} &:= \Id_{\Field}:L_1^0\to L_2^0\\
    B^\odot|_{L^{-q}} &:= (B^{-1})^*\otimes \stackrel{q}{\cdots} \otimes (B^{-1})^* :L_1^{-q}\to L_2^{-q}
\end{align}
where we have crucially used the invertibility of the factor $B$. By construction, $B^\odot$ is compatible with the $\Int$-dimensioned structure and since $B$ is a linear map with linear inverse, all the tensor powers act as $\Field$-linear maps on the dimension slices, thus making $B^\odot:L_1^\odot \to L_2^\odot$ into a morphism of abelian dimensioned groups. Showing that $B^\odot$ is a dimensioned ring morphism follows easily by the explicit construction of the dimensioned ring multiplication $\odot$ given in proposition \ref{DimRingPower} above. This is checked directly for products that do not mix positive and negative tensor powers and for mixed products it suffices to note that
\begin{equation}
    B^\odot (\alpha)\odot B^\odot (a)=(B^{-1})^*(\alpha) \odot B(a)= \alpha (B^{-1}(B(a)))= \alpha(a) = \Id_{\Field}(\alpha(a)) =B^\odot (\alpha\odot a).
\end{equation}
It follows from the usual properties of tensor products in vector spaces that for another factor $C:L_2\to L_3$ we have
\begin{equation}
    (C\circ B)^\odot=C^\odot \circ B^\odot, \qquad (\Id_L)^\odot=\Id_{L^\odot},
\end{equation}
thus making the power assignment into a functor. Recall that a choice of unit in a line $L\in \Line$ is simply a choice of non-vanishing element $u\in L^\times$. In proposition \ref{DimRingPower} we saw that $L^\odot$ is a dimensioned field, so multiplicative inverses exist, let us denote them by $u^{-1}\in (L^*)^\times$. Using the notation for $q>0$
\begin{align}
    u^q &:= u\odot \stackrel{q}{\cdots} \odot u\\
    u^0 &:= 1\\
    u^{-q} &:= u^{-1}\odot \stackrel{q}{\cdots} \odot u^{-1},
\end{align}
it is clear that the map
\begin{align}
u: \Int & \to L^\odot\\
n & \mapsto u^n
\end{align}
satisfies
\begin{equation}
    u^{n+m}=u^n\odot u^m.
\end{equation}
By construction, all $u^n\in L^n$ are non-zero, so $u:\Int \to L^\odot$ is a choice of units in the dimensioned field $(L^\odot_\Int,+_\Int,\odot)$. The isomorphism of dimensioned fields $L^\odot \cong \Field \times \Int$ follows from proposition \ref{UnitsDimFields} and the observation that, by definition, $(L^\odot)_0=L^0=\Field$.
\end{proof}

The power ring construction of a line can be easily extended to a collection of lines: given the ordered set of lines $L_1,\dots,L_k\in\Line$, we define their \textbf{power ring} as:
\begin{equation}
    (L_1,\dots,L_k)^\odot:=\bigcup_{n_1,\dots n_k\in \Int} L_1^{n_1}\otimes \cdots \otimes L_k^{n_k},
\end{equation}
which has a natural abelian dimensioned group structure given by $\Field$-linear addition and has dimension group $\Int^k$. The dimensioned multiplicative structure generalizes in the obvious way:
\begin{equation}
    (a_1\otimes \dots \otimes a_k) \odot (b_1\otimes \dots \otimes b_k):= a_1\odot b_1 \otimes \dots \otimes a_k\odot b_k
\end{equation}
thus making $((L_1,\dots,L_k)^\odot_{\Int^k},+_{\Int^k},\odot)$ into a dimensioned field. Note that the powers of each individual line $L_i$ can be found as dimensioned subfields $L_i^\odot\subset (L_1,\dots,L_k)^\odot$ since they are simply the dimensional preimages of the natural subgroups $\Int\subset \Int^k$. Furthermore, a choice of unit in each of the individual lines $u_i\in L_i^\times$ naturally induces a choice of units for the power ring in a natural way
\begin{align}
U: \Int^k & \to (L_1,\dots,L_k)^\odot\\
(n_1,\dots,n_k) & \mapsto u_1^{n_1}\odot \cdots \odot u_k^{n_k},
\end{align}
which, in turn, gives the dimensioned ring isomorphism:
\begin{equation}
    (L_1,\dots,L_k)^\odot \cong \Field \times \Int^k.
\end{equation}

We thus recover the explicit algebraic structure of physical quantities identified in Section \ref{physicalquantity} from the standard practice in dimensional analysis. This concludes our initial goal to develop a rigorous mathematical framework where the formal structure of physical quantities is naturally accounted for.

\section{Dimensioned Modules} \label{dimmod}

Ordinary modules are algebraic structures closely related to rings: there is an internal additive structure and multiplication is defined externally in such a way that properties, as formally close as possible to the ring axioms, are satisfied. We motivate the definition of dimensioned modules by investigating the structure present in natural constructions with dimensioned rings.\newline

Considering the dimensional abelian group part of a dimensioned ring $(R_D,+_D,\cdot^D)$, we can form the product $R_D\times R_D$, which is a dimensional abelian group with dimension set $D\times D$, or the direct sum $R_D \oplus_D R_D$, which is a dimensional abelian group with dimension set $D$. In both cases we can form module-like maps by setting
\begin{equation}
    a_d * (b_e,c_f):=(a_d\cdot b_e,a_d\cdot c_f), \qquad a_d *(b_e\oplus c_e):= a_d\cdot b_e \oplus a_d\cdot c_e.
\end{equation}
These module-like actions are compatible with the dimensioned structure in the sense that, in the first case, $D$ acts diagonally on $D\times D$ and, in the second case, $D$ acts on itself by multiplication. Furthermore, from the defining axioms of dimensioned ring, these maps satisfy the usual linearity properties of the conventional notion of $R$-module with the only caveat that addition is partially defined.\newline

Recall from our discussion in Section \ref{dimring} that the dimensioned maps from $R_D$ into itself form an abelian dimensioned group $(\Dim{R_D}_{\text{Map}(D)},+_{\text{Map}(D)})$ where $\text{Map}(D)$ denotes the set of maps from $D$ onto itself. The presence of the dimensioned ring multiplication allows for the definition of the following module-like structure
\begin{equation}
    * :R_D\times \Dim{R_D} \to \Dim{R_D}
\end{equation}
defined via
\begin{equation}
    (a_d * \Phi)(b_e):=a_d\cdot \Phi(b_e).
\end{equation}
We note that $a_d * \Phi$ is a well-defined dimensioned morphism from the fact $D$ acts naturally on $\text{Map}(D)$ by composition with the monoid multiplication action of $D$ on itself: indeed the if $\phi:D\to D$ is the dimension map of $\Phi$, then $a_d * \Phi$ has dimension map $d\circ \phi:D\to D$. Once more, it follows directly from the axioms of dimensioned ring that that this operation satisfies the usual linearity properties of the conventional notion of $R$-module with the only caveat that addition is partially defined.\newline

These examples motivate the following definition: let $(R_G,+_G,\,^G)$ be a dimensioned ring (in  the  interest  of  notational  economy,  ring  multiplications  will be  denoted by  juxtaposition  hereafter) and $(A_D,+_D)$ a dimensioned abelian group. Note that $G$ carries a monoid structure whereas $D$ is simply a set. $A_D$ is called a \textbf{dimensioned $R_G$-module} if there is a map 
\begin{equation}
    \cdot :R_G\times A_D \to A_D
\end{equation}
that is compatible with the dimensioned structures via a monoid action $G\times D\to D$ (denoted by juxtaposition) in the following sense
\begin{equation}
    r_g\cdot a_d=(r \cdot a)_{gd}
\end{equation}
and that satisfies the following axioms
\begin{itemize}
    \item[1)] $r_g\cdot (a_d+b_d)=r_g\cdot a_d + r_g\cdot b_d$,
    \item[2)] $(r_g+p_g)\cdot a_d=r_g\cdot a_d + p_g \cdot a_d$,
    \item[3)] $(r_g p_h)\cdot a_d=r_g\cdot (p_h\cdot a_d)$,
    \item[4)] $1\cdot a_d=a_d$
\end{itemize}
for all $r_g,p_h\in R_G$ and $a_d,b_d\in A_D$. Note that these four axioms for a map $\cdot :R_G\times A_D \to A_D$ can only be demanded in consistency with the dimensioned structure in the presence of a monoid action $G\times D\to D$. With this definition at hand, we recover the motivating examples: the direct sum $R_G \oplus_G R_G$ is a dimensioned $R_G$-module with dimension set $G$ and monoid action given by the multiplication action; the product $R_G\times R_G$ is a dimensioned $R_G$-module with dimension set $G\times G$ and monoid action given by the diagonal action; and the set of dimensioned maps of a dimensioned ring $\Dim{R_G}$ is a dimensioned $R_G$-module with dimension set $\text{Map}(G)$ and monoid action given by composition with the multiplication action.\newline

Let $(A_D,+_D)$ and $(B_E,+_E)$ be two dimensioned $R_G$-modules, a morphism of abelian dimensioned groups $\Phi:A_D\to B_E$ is called \textbf{$R_G$-linear} if
\begin{equation}
    \Phi(r_g\cdot a_d)=r_g\cdot \Phi(a_d)
\end{equation}
for all $r_g\in R_G$ and $a_d\in A_D$. Note that this condition forces the dimension map $\phi:D\to E$ to satisfy
\begin{equation}
    \phi(gd)=g\phi(d)
\end{equation}
for all $g\in G$ and $d\in D$, in other words, the dimension map $\phi$ must be $G$-equivariant with respect to the monoid actions of the dimension sets $D$ and $E$. Let us denote the set of $G$-equivariant dimension maps as
\begin{equation}
    \text{Map}^G(D,E):=\{\phi:D\to E \, | \quad \phi\circ g = g \circ \phi \quad \forall \, g\in G\},
\end{equation}
then it follows that the dimensioned group of morphisms $\Dim{A_D,B_E}_{\text{Map}(D,E)}$ contains a dimensioned subgroup of morphisms covering $G$-equivariant dimension maps for which the following dimensioned module map can be defined
\begin{equation}
    (r_g\cdot \Phi)(a_d):=r_g\cdot \Phi(a_d)=\Phi(r_g\cdot a_d).
\end{equation}
The set of dimensioned maps $\Dim{A_D,B_E}_{\text{Map}^G(D,E)}\subset \Dim{A_D,B_E}_{\text{Map}(D,E)}$ that are $R_G$-linear is thus shown to carry a natural dimensioned $R_G$-module structure. We simply call these the \textbf{$R_G$-linear maps} between $A_D$ and $B_E$ and denote them by $\text{Dim}_{R_G}(A_D,B_E)$. Dimensioned $R_G$-modules together with $R_G$-linear maps form a category denoted by $\normalfont R_G\textsf{DimMod}$.\newline

Let $(A_D,+_D)$ be a dimensioned $R_G$-module, a dimensional abelian subgroup $S\subset A_D$ is called a \textbf{dimensioned submodule} if
\begin{equation}
    r_g\cdot s_d\in S\cap A_{gd}
\end{equation}
for all $r_g\in R_G$ and $s_d\in S$. Natural examples of dimensioned submodules are the \textbf{span} of a subset $X\subset A_D$, defined as all the possible $R_G$-linear combinations of elements in $X$, and the kernels and images of $R_G$-linear maps between modules. The dimensional abelian group quotient construction of Section \ref{dimring} induces the notion of \textbf{quotient} of $R_G$-modules: let $\delta: A\to D$ be the dimension projection of the dimensioned $R_G$-module $(A_D,+_D)$ and $S\subset A_D$ a submodule, then by taking the quotient as dimensional abelian groups $A'_{\delta(S)}:= A_D/S$ is a dimensioned $R_G$-module.\newline

Let $(A_D,+_D)$ and $(B_D,+_D)$ be two dimensioned $R_G$-modules, the dimensional abelian group direct sum $A_D\oplus_D B_D$ carries a natural $R_G$-module structure:
\begin{equation}
    r_g\cdot (a_d\oplus_d b_d) := r_g\cdot a_d \oplus_{gd} r_g\cdot b_d,
\end{equation}
which gives the definition of \textbf{direct sum} of dimensioned $R_G$-modules. Our definitions so far allow for notions from ordinary module theory, such as \textbf{finitely generated}, \textbf{free}, \textbf{projective} or \textbf{injective}, to apply to dimensioned modules in an obvious way. By fixing a dimensioned ring $R_G$ and a dimension set $D$, $R_G$-modules with dimensions in $D$ together with $D$-preserving $R_G$-linear maps form an abelian category, essentially analogous to the category of $D$-dimensional abelian groups $\textsf{DimAb}_D$. This is called the \textbf{category of $D$-dimensional $R_G$-modules} denoted by $R_G\textsf{DimMod}_D$.\newline

Let $(A_D,+_D)$ and $(B_E,+_E)$ be two dimensioned $R_G$-modules, we define their \textbf{tensor product} from the product of the underlying dimensional abelian groups and quotienting by the dimensioned ring action:
\begin{equation}
    A_D\otimes_{R_G} B_E := (A_D \otimes B_E)/\sim_{R_G}
\end{equation}
where $\sim_{R_G}$ is defined in the obvious way:
\begin{equation}\label{tensorquorel}
    (r_g\cdot a_d,b_e)\sim_{R_G} (a_d,r_g\cdot b_e)
\end{equation}
for all $r_g\in R_G$, $a_d\in A_D$ and $b_e\in B_E$. By construction, $A_D\otimes_{R_G} B_E$ is a dimensional abelian group with dimension set given by a peculiar quotient of the product of dimension sets:
\begin{equation}
    D \times^G E := (D \times E)/\sim_G
\end{equation}
where $\sim_G$ is induced from $\sim_{R_G}$ above in the obvious way:
\begin{equation}
    (g d,e)\sim_G (d,g e),
\end{equation}
for all $g\in G$, $d\in D$, $e\in E$, making use of the monoid actions. These constructions ensure that $A_D\otimes_{R_G} B_E$ is indeed a $R_G$-module whose action can be made explicit by using the element-wise tensor product notation as in the case of tensor product of ordinary modules:
\begin{equation}
    (r_g\cdot a_d) \otimes b_e = r_g \cdot (a_d \otimes b_e) = a_d \otimes (r_g \cdot b_e).
\end{equation}
It follows that the tensor product construction can be characterised with the obvious universal property, defining in turn the tensor product of $R_G$-linear maps and establishing \textbf{$R_G$-bilinearity} of a map $\Phi: A_D\times B_E\to C_F$ as the fact that it factors through the tensor product via $R_G$-linear maps:
\begin{equation}
\begin{tikzcd}
A_D\times B_E \arrow[r, "\otimes"] & A_D\otimes_{R_G} B_E \arrow[r, "\phi"] & C_F
\end{tikzcd}
\end{equation}
where $\otimes$ denotes the natural element-wise tensor product as a dimensioned map covering the canonical projection $\eta: D\times E\to D \times^G E$.\newline

In complete analogy with the tensor product of ordinary modules, it is easy to check that the category of dimensioned $R_G$-modules $R_G \textsf{DimMod}$ forms a symmetric monoidal category with the above tensor product and $R_G$ as the unit object. The novelty of our definition lies in the fact that the dimension sets also inherit a monoidal structure. Consider the category of sets as objects and $G$-equivariant maps as morphisms\footnote{Technically, a morphism is a triple specifying $G$-actions on a pair of sets and an equivariant map between them.}; we denote this category by $\Set^G$. It is easy to see that the product $\times^G$ defined above as the dimension component of the tensor product of dimensioned modules is a well-defined tensor product in its own right making $(\Set^G,\times^G)$ into a monoidal category with $G$ as the unit object ($G$ acts on itself by monoid multiplication). Identifying this category allows to recast our definition of the category of $R_G$-modules in a more evocative fashion: the definitions of $R_G$-linearity and the dimension structures of all the sets involved imply that there is a full surjective functor the form:
\begin{equation}
\begin{tikzcd}
R_G\textsf{DimMod} \arrow[d, "\textsf{d}"]\\
\Set^G
\end{tikzcd}
\end{equation}
This invites us to think of $R_G\textsf{DimMod}$ as a category with dimensions on the category $\Set^G$.\newline

The theory of dimensioned modules developed thus far is manifestly analogous to the theory of ordinary modules; all the standard notions of ordinary modules appear identical aside from the \emph{dimensioned technology} present to systematically account for the partial nature of addition. Proposition \ref{DimCat} below further vindicates this statement by showing that dimensioned modules form the dimensioned analogue of a \textbf{rig category} \cite{baez2010rig} \cite{johnson2021bimonoidal}, which, as is the case for ordinary modules over a ring, amounts to the presence of two monoidal structures interacting precisely as the binary operations of a rig (a ring without additive inverses).

\begin{prop}[Dimensioned Rig Category] \label{DimCat}
The direct sum $\oplus$ on $\normalfont R_G\textsf{DimMod}$ is a dimensional binar with dimensions in $\normalfont \Set^G$ and the tensor product $\otimes$ on $\normalfont R_G\textsf{DimMod}$ is a dimensioned binar compatible with the monoid product of $\normalfont \Set^G$. Furthermore, $\otimes$ is distributive with respect to $\oplus$ within the category $\normalfont R_G\textsf{DimMod}$ thus making
\begin{equation}
\normalfont
    (R_G\textsf{DimMod}_{\Set^G}, \oplus_{\Set^G},\otimes^{\Set^G})
\end{equation}
into a \textbf{dimensioned rig category}.
\end{prop}
\begin{proof}
Recall that a dimensional binar structure is simply an ordinary binar structure restricted to the preimages of the dimension projection, which, in the case of the $\textsf{d}$ functor above, are simply the categories of $D$-dimensional $R_G$-modules $R_G\textsf{DimMod}_D$ for some fixed set $D\in \Set^G$. The categories $R_G\textsf{DimMod}_D$ are easily shown to be abelian categories with respect to the direct sum restricted to the fixed set $\oplus_D$. The tensor product of $R_G$-modules acts as a dimensioned operation on $R_G\textsf{DimMod}$ by construction, as indeed one has:
\begin{equation}
    A_D \otimes_{R_G} B_E = (A \otimes_{R_G} B)_{D \times^G E}.
\end{equation}
The additive identity is the trivial module $0_1$ whose underlying dimensioned set is the singleton surjection ${0}\to 1$, it then follows from the quotient relation (\ref{tensorquorel}) that there are natural isomorphisms:
\begin{equation}
    A_D\otimes_{R_G} 0_1 \cong 0_1 \qquad 0_1 \otimes_{R_G} A_D \cong 0_1
\end{equation}
thus giving the usual absorption identities of a rig. Similarly, it follows from the defining slice-wise relations of dimensional free abelian group construction that the tensor product is right distributive with respect to the direct sum:
\begin{equation}
    (A_D \oplus_D B_D )\otimes_{R_G} C_E \cong A_D \otimes_{R_G} C_E \oplus_{D\times^G E} B_D \otimes_{R_G} C_E
\end{equation}
and similarly for left distributivity.
\end{proof}

So far we have only considered dimensioned modules over the same dimensioned ring. We can connect categories of dimensioned modules over different dimensioned rings via the \textbf{pullback} construction: let $(A_D,+_D)$ be a $R_G$-module and $\varphi: P_H\to R_G$ a dimensioned ring morphism, then $A_D$ has a dimensioned $P_H$-module structure given by:
\begin{equation}
    p_h\cdot a_d := \varphi(p_h)\cdot a_d
\end{equation}
for all $p_h\in P_H$ and $a_d\in A_D$. This dimensioned module is denoted by $\varphi^*A_D$ since the base set of the module is unchanged and so is its dimension set; the monoid action $H\times D \to D$ is given by pullback with the dimension map $H\to G$ of the dimensioned ring morphism $\varphi: P_H\to R_G$. This construction motivates the extension of the notion of dimensioned module morphisms to account for maps between modules over different rings: let $A_D$ be a $R_G$-module and $B_E$ a $P_H$-module, the pair of maps $\Phi^\varphi$ is said to be a \textbf{twisted module morphism} if $\Phi_F: A_D\to B_E$ is a dimensioned map, $\varphi_f: R_G\to P_H$ is a dimensioned ring morphism and
\begin{equation}
    \Phi(r_g\cdot a_d)=\varphi(r_g)\cdot \Phi(a_d)
\end{equation}
for all $r_g\in R_G$ and $a_d\in A_D$. We also say that $\Phi: A_D\to B_E$ is a \textbf{$\varphi$-linear map}. Note that the $\varphi$-linearity condition implies that the dimension map $F:D\to E$ satisfies a twisted equivariance property with respect to the monoid actions:
\begin{equation}
    F(gd)=f(g)F(d)
\end{equation}
where $g\in G$, $d\in D$. For two twisted module morphisms $\Phi^\phi$ and $\Psi^\psi$ it is easily checked that:
\begin{equation}
    \Phi^\phi \circ \Psi^\psi = (\Phi \circ \Psi)^{\phi \circ \psi}
\end{equation}
and so the categories of dimensioned modules over a fixed dimensioned ring $R_G\textsf{DimMod}$ can now be generalised to include all module morphisms twisted by the endomorphisms of $R_G$.

\begin{prop}[Pullback Functor] \label{PullBackFunct}
Let $\varphi: P_H\to R_G$ be dimensioned ring morphism, then the pullback construction is compatible with composition and categorical products. Furthermore, when $\varphi$ is an isomorphism, the assignment
\begin{equation}
\normalfont
    \varphi^* : R_G\textsf{DimMod} \to P_H\textsf{DimMod}.
\end{equation}
becomes a functor of dimensioned rig categories, the categorical analogue of a dimensioned ring morphism.
\end{prop}
\begin{proof}
Note that the pullback assignment $\varphi^*$ is the identity functor at the level of dimensional abelian groups by construction. This means that we should only check compatibility with the pullback module multiplication structure. Compatibility with a $\psi$-linear map of $R_G$-modules $\Psi^\psi:A_D\to B_E$ is ensured by setting:
\begin{equation}
    \varphi^*(\Psi^\psi):= \Psi^{\psi \circ \varphi}
\end{equation}
which, by simple checks is shown to satisfy:
\begin{equation}
    \varphi^*(\Id_{A_D}^{\Id_{R_G}})=\Id_{\varphi^*A_D}^{\Id_{P_H}} \qquad \varphi^*(\Phi^\phi \circ \Psi^\psi)=\varphi^*\Phi^\phi \circ \varphi^*\Psi^\psi.
\end{equation}
Further relying on the underlying dimensional abelian group structure of dimensioned modules, we can easily show that pullbacks preserve categorical products in the following sense:
\begin{equation}
    \varphi^*(A_D\oplus_D B_D)\cong \varphi^*A_D \oplus_D \varphi^*B_D \qquad \varphi^*(A_D\otimes_{R_G} B_E)\cong \varphi^*A_D \otimes_{P_H} \varphi^*B_E 
\end{equation}
where the isomorphisms are as dimensioned modules. Note that additivity and multiplicativity of the dimensioned ring morphism $\varphi$ is crucial for the construction of these isomorphisms. The pullback assignment defined in this way sends objects to objects between the categories of dimensioned modules over the fixed rings $R_G$ and $P_H$, but it does not so for morphisms of those categories since $\varphi^*(\Psi^\psi):\varphi^*A_D\to B_E$ is a map from a $P_H$-module to a $R_G$-module. When $\varphi$ is an isomorphism, however, we can define the pullback assignment making use of the inverse:
\begin{equation}
    \varphi^*(\Psi^\psi):= \Psi^{\varphi^{-1}\circ \psi \circ \varphi}
\end{equation}
which then makes it into a well defined functor.
\end{proof}

An important situation where twisted module morphisms appear naturally is the general notion of quotient of modules induced by ideal submodules.

\begin{prop}[Quotient Dimensioned Module] \label{QuotMod}
Let $(A_D,+_D)$ be a dimensioned $R_G$-module and $I\subset R_G$ an ideal, then, if $S\subset A_D$ is a submodule such that $I\cdot A\subset S$, the quotient $A_D/S$ inherits a dimensioned $R_G/I$-module structure such that the projection map
\begin{equation}
    Q: A_D\to A_D/S
\end{equation}
is a $q$-linear map, where $q:R_G\to R_G/I$ is the quotient ring projection.
\end{prop}
\begin{proof}
Since quotients of rings and modules are taken as dimensional abelian groups, this result follows from a simple computation showing the distributivity property of the ideal submodule $S$: let $r_g\in R_G$, $i_g\in I$, $a_d\in A_D$ and $s_d\in S$ then
\begin{equation}
    (r_g+i_g)\cdot (a_d+s_d) = r_g \cdot a_d + r_g \cdot s_d + i_g \cdot a_d + i_g \cdot s_d.
\end{equation}
The second and fourth terms are in $S$ from the fact that $S$ is a submodule and the third term is in $S$ from the ideal submodule condition $I\cdot A\subset S$, then the above expression defines the $R_G/I$-module structure on $A_D/S$ which, by construction, satisfies $Q(r_g\cdot a_d)=q(r_g)\cdot Q(a_d)$.
\end{proof}

\section{Dimensioned Algebras} \label{dimalg}

Let us motivate the dimensioned generalization of the notion of algebra by considering, once more, the guiding example of the set of dimensioned maps of a dimensioned ring $R_G$. Proposition \ref{DimEndRing} implies that the set of dimensioned maps of $R_G$, regarded as a dimensional abelian group, carries a dimensioned ring structure $(\Dim{R_G}_{\text{Map}(G)},+_{\text{Map}(G)},\circ^{\text{Map}(G)})$; but $\Dim{R_G}_{\text{Map}(G)}$ is also naturally a $R_G$-module structure. Depending on which further conditions are imposed on the dimensioned maps, these two structures, the dimensioned ring and the $R_G$-module, may interact in different ways. A first obvious choice is to consider dimensioned ring homomorphisms, in which case the interaction manifests as the fact that composition acts as a twisted $R_G$-module morphism. Another direction is to consider differential operators. Although we will only focus on zeroth order operators and derivations, general differential operators are defined recursively from the $R_G$-linearity condition:
\begin{equation}
    \Phi_\phi(r_gs_h)=r_g\Phi_\phi(s_h)=\Phi_\phi(r_g)s_h
\end{equation}
for $r_g,s_h\in R_G$ and $\Phi_\phi:R_G \to R_G$ a dimensioned map. In the case at hand of commutative dimensioned rings, this condition can only be realised by multiplication by a ring element, which corresponds to the natural inclusion $R_G\hookrightarrow \text{Dim}_{R_G}(R_G)$ via ring multiplication. Dimension maps of such $R_G$-linear operators correspond, in turn, to multiplication by monoid elements. It is then easy to see that $R_G$-linear operators satisfy
\begin{equation}
    r_g\cdot (\Phi_\phi \circ \Psi_\psi)= (r_g\cdot \Phi_\phi) \circ \Psi_\psi=\Phi_\phi \circ (r_g\cdot \Psi_\psi)
\end{equation}
for all $r_g\in R_G$ and $\Phi_\phi, \Psi_\psi\in \text{Dim}_{R_G}(R_G)$ . This shows that $\text{Dim}_{R_G}(R_G)$ gives a prime example of a bilinear associative operation on a dimensioned module and prompts us to give the following general definition.\newline

Let $(A_D,+_D)$ be a dimensioned $R_G$-module, a map $M:A_D\times A_D\to A_D$ is called a \textbf{dimensioned bilinear multiplication} if it satisfies
\begin{align}
    M(a_d \, +_d\, b_d,c_e)&=M(a_d,c_e) \,+_{\mu(d,e)} \, M(b_d,c_e)\\
    M(a_d,b_e \, +_e \, c_e)&=M(a_d,b_e) \, +_{\mu(d,e)} \, M(a_d,c_e)\\
    M(r_g\cdot a_d,s_h\cdot b_e)&=r_g\cdot s_h\cdot M(a_d, b_e)
\end{align}
for all $a_d,b_d,b_e,c_e\in A_D$, $r_g,s_h\in R_G$ and for a \textbf{dimension map} $\mu:D\times D\to D$ which is $G$-equivariant in both entries, i.e.
\begin{equation}
    \mu(gd,he)=gh\mu(d,e)
\end{equation}
for all $g,h\in G$ and $d,e\in D$. When such a map $M$ is present in a dimensioned $R_G$-module $A_D$, the pair $(A_D,M)$ is called a \textbf{dimensioned $R_G$-algebra}. Note that the tensor product of dimensioned modules given a the end of Section \ref{dimmod} allows to reformulate the definition of a dimensioned bilinear multiplication $M:A_D\times A_D\to A_D$ as a dimensioned $R_G$-linear morphism
\begin{equation}
    M:A_D\otimes_{R_G} A_D \to A_D.
\end{equation}

The natural notions of morphisms and subalgebras of ordinary algebras extend naturally to the dimensioned case. Let $(A_D,M)$ and $(B_E,N)$ be two dimensioned $R_G$-algebras, a $R_G$-linear morphism $\Phi:A_D\to B_E$ is called a \textbf{morphism of dimensioned algebras} if
\begin{equation}
    \Phi(M(a,a'))=N(\Phi(a),\Phi(a')),
\end{equation}
for all $a,a'\in A_D$. A submodule $S\subset A_D$ such that $M(S,S)\subset S$ is called a \textbf{dimensioned subalgebra}.\newline

The dimension map $\mu$ of a dimensioned bilinear multiplication in a dimensioned $R_G$-algebra $(A_D,M_\mu)$ is naturally regarded as an binary operation on the set of dimensions $D$, i.e. a binar. These are, generally, featureless algebraic structures, however, if one wishes to demand specific algebraic properties, such as commutativity or associativity, the algebraic structure present in the dimension binar becomes necessarily richer. Let $(A_D,M_\mu)$ be a dimensioned $R_G$-algebra, we say that it is \textbf{symmetric} or \textbf{antisymmetric} if
\begin{equation}
    M(a_d,b_e)=M(b_e,a_d), \qquad M(a_d,b_e)=-M(b_e,a_d)
\end{equation}
for all $a_d,b_e\in A_D$, respectively. The dimension binars of symmetric or antisymmetric dimensioned algebras are necessarily commutative, i.e. $\mu(d,e)=\mu(e,d)$ for all $d,e\in D$. The usual 3-element-product properties of ordinary algebras can be demanded for dimensioned algebras in an analogous way, in particular $(A_D,M_\mu)$ is called \textbf{associative} or \textbf{Jacobi} if
\begin{equation}
    \text{Assoc}_M(a_d,b_e,c_f)=0,\qquad \text{Jac}_M(a_d,b_e,c_f)=0
\end{equation}
for all $a_d,b_e,c_f\in A_D$, respectively. The dimension binars of associative or Jacobi dimensioned algebras are necessarily associative, i.e. $\mu(\mu(d,e),f)=\mu(d,\mu(e,f))$ for all $d,e,f\in D$, making them into semigroups. Returning to the motivating example presented at the beginning of this section, we now see that the dimensioned morphisms of a dimensioned ring $R_G$ give the prime example of dimensioned associative algebra $(\Dim{R_G},\circ)$.\newline

In parallel with the definitions of ordinary algebras, we define \textbf{dimensioned commutative algebra} as a symmetric and associative dimensioned algebra and a \textbf{dimensioned Lie algebra} as an antisymmetric and Jacobi dimensioned algebra. Note that dimensioned commutative and dimensioned Lie algebras necessarily carry dimension sets that are commutative semigroups.\newline

In keeping with the general philosophy to continue to scrutinize the natural algebraic structure present in the dimensioned module of dimensioned maps of a dimensioned ring $R_G$, let us attempt to find the appropriate dimensioned generalization of the notion of derivations of a ring. Working by analogy, a dimensioned derivation will be a dimensioned map $\Delta\in \Dim{R_G}$ covering a dimension map $\delta:G\to G$ satisfying a Leibniz identity with respect to the dimensioned ring multiplication
\begin{equation}
    \Delta(r_g\cdot s_h)= \Delta(r_g)\cdot s_h + r_g\cdot\Delta(s_h),
\end{equation}
for all $r_g,s_h\in R_G$, however, for the right-hand-side to be well-defined, both terms must be of homogeneous dimension, which means that the dimension map must satisfy
\begin{equation}
    \delta(gh)=\delta(g)h=g\delta(h)
\end{equation}
for all $g,h\in G$. Since $G$ is a monoid, this condition is equivalent to the dimension map being given by multiplication with a monoid element, i.e. $\delta=d$ for some element $d\in G$. Following from this observation, we see that there is a natural dimensioned submodule of the dimensioned module of dimensioned maps $\Dim{R_G}_G\subset \Dim{R_G}_{\text{Map(G)}}$ given by the dimensioned maps whose dimension maps are specified by multiplication with a monoid element. Recall that dimensioned rings are assumed to be commutative and, thus, the dimension monoid has commutative binary operation. This allows for the identification of the first natural example of dimensioned Lie algebra: the commutator of the associative dimensioned composition
\begin{equation}\label{dimcomm}
    [\Delta,\Delta']:=\Delta \circ \Delta'-\Delta' \circ \Delta,
\end{equation}
is easily checked to be antisymmetric and Jacobi, thus making $(\Dim{R_G}_G,[,])$ into the \textbf{dimensioned Lie algebra of dimensioned maps} of a dimensioned ring $R_G$.\newline

Notice that commutator bracket (\ref{dimcomm}) can only be defined on the dimensioned submodule $\Dim{R_G}_G\subset \Dim{R_G}_{\text{Map(G)}}$ since the two terms of the right-hand-side will have different dimensions in general since the composition of maps from $G$ into itself is a non-commutative binary operation. This motivates the definition of the dimensioned Lie algebra of \textbf{derivations of a dimensioned ring} $R_G$ as the natural dimensioned Lie subalgebra of the dimensioned maps:
\begin{equation}
    \Dr{R_G}\subset (\Dim{R_G}_G,[,]).
\end{equation}
Derivations covering the identity dimension map $\Id_G:G\to G$ are called \textbf{dimensionless derivations} and it is clear by definition that they form an ordinary Lie algebra with the commutator bracket $(\Dr{R_G}_1, [,])$. By restricting their action to elements of the dimensionless ring $R_1\subset R_G$ we recover the ordinary Lie algebra of ring derivations, in other words, there is a surjective map of Lie algebras
\begin{equation}
    (\Dr{R_G}_1,[,])\to (\Dr{R_1},[,]).
\end{equation}

The example of the dimensioned Lie algebra of derivations illustrates the case of a dimensioned algebra whose set of dimensions is a (commutative) monoid and whose dimension map is simply given by the monoid multiplication. For the remainder of this work, the dimension sets of dimensioned modules will be assumed to carry a commutative monoid structure (with multiplication denoted by juxtaposition of elements) unless stated otherwise. Let $A_G$ be a dimensioned module, a dimensioned algebra multiplication $M:A_G\times A_G\to A_G$ is said to be \textbf{homogeneous of dimension $m$} if the dimension map $\mu:G\times G\to G$ is given by monoid multiplication with the element $m\in G$, i.e. $\mu(g,h)=mgh$ for all $g,h\in G$. Assuming a monoid structure on the dimension set of a dimensioned module and considering dimensioned algebra multiplications of homogeneous dimension is particularly useful in order to study several algebra multiplications coexisting on the same set. Indeed, given two homogeneous dimensioned algebra multiplications $(A_G,M_1)$ and $(A_G,M_2)$ with dimensions $m_1\in G$ and $m_2\in G$, respectively, the fact that the monoid operation is assumed to be associative and commutative, allows for consistently demanding properties of the interaction of the two dimensioned multiplications involving expressions of the form $M_1(M_2(a,b),c)$ without any further requirements.

\section{Dimensioned Poisson Algebras} \label{dimPoisson}

We now put all the dimensioned algebra machinery to work aiming to develop the dimensioned analogue of Poisson algebras. Let $A_G$ be a dimensioned $R_H$-module and let two dimensioned algebra multiplications $*:A_G\times A_G\to A_G$ and $\{,\}:A_G\times A_G\to A_G$ with homogeneous dimensions $p\in G$ and $b\in G$, respectively, the triple $(A_G,*_p,\{,\}_b)$ is called a \textbf{dimensioned Poisson algebra} if
\begin{itemize}
    \item[1)] $(A_G,*_p)$ is a dimensioned commutative algebra,
    \item[2)] $(A_G,\{,\}_b)$ is a dimensioned Lie algebra,
    \item[3)] the two multiplications interact via the Leibniz identity
    \begin{equation}
        \{a,b*c\}=\{a,b\}*c+b*\{a,c\},
    \end{equation}
    for all $a,b,c\in A_G$.
\end{itemize}
Note that the Leibniz condition can be consistently demanded of the two dimensioned algebra multiplications since the dimension projections of each of the terms of the Leibniz identity for $\{a_g,b_h*c_k\}$ are: 
\begin{equation}
    bgphk,\quad pbghk,\quad phbgk,
\end{equation}
which are all indeed equal from the fact that the monoid binary operation is associative and commutative.\newline

A morphism of dimensioned modules between dimensioned Poisson algebras $\Phi:(A_G,*_p,\{,\}_b)\to (B_H,*_r,\{,\}_c)$ is called a \textbf{morphism of dimensioned Poisson algebras} if $\Phi:(A_G,*_p)\to (B_H,*_r)$ is a morphism of dimensioned commutative algebras and also $\Phi:(A_G,\{,\}_b)\to (B_H,\{,\}_c)$ is a morphism of dimensioned Lie algebras. A submodule $I\subset A_G$ that is a dimensioned ideal in $(A_G,*_p)$ and that is a dimensioned Lie subalgebra in $(A_G,\{,\}_b)$ is called a \textbf{dimensioned coisotrope}. The category of dimensioned Poisson algebras is denoted by $\textsf{DimPoissAlg}$.\newline

We conclude our discussion of dimensioned algebra by showing that the category of dimensioned Poisson algebras admits the all-important constructions of reductions and products, which turn out to be direct generalizations of their ordinary counterparts for Poisson algebras.

\begin{prop}[Dimensioned Poisson Reduction] \label{DimensionedPoissonReduction}
Let $(A_G,*_p,\{,\}_b)$ be a dimensioned Poisson algebra and $I\subset A_G$ a coisotrope, then there is a dimensioned Poisson algebra structure induced in the subquotient
\begin{equation}
    (A'_G:=N(I)/I,*'_p,\{,\}'_b)
\end{equation}
where $N(I)$ denotes the dimensioned Lie idealizer of $I$ regarded as a submodule of the dimensioned Lie algebra.
\end{prop}
\begin{proof}
We assume without loss of generality that the dimension projection of $I$ is the whole of $G$, the intersections with the dimension slices are denoted by $I_g:=I\cap A_g$. The dimensioned Lie idealizer is defined in the obvious way:
\begin{equation}
    N(I):=\{n_g\in A_G|\quad \{n_g,i_h\}\in I_{bgh} \quad \forall i_h\in I\}.
\end{equation}
Note that $N(I)$ is the smallest dimensioned Lie subalgebra that contains $I$ as a dimensioned Lie ideal. The Leibniz identity implies that $N(I)$ is a dimensioned commutative subalgebra with respect to $*_p$ in which $I$ sits as a dimensioned commutative ideal, since it is a commutative ideal in the whole $A_G$. Following a construction analogous to the quotient dimensioned ring of Proposition \ref{QuotDimRing} we can form the dimensioned quotient commutative algebra $(N(I)/I,*')$ whose underlying module is the quotient dimensioned module of Proposition \ref{QuotMod}. Note that the only difference with the case of the quotient dimensioned ring construction is that the commutative multiplication $*_p$ covers a dimension map that is given by the monoid multiplication with a non-identity element $p\in G$, but this has no effect on the quotient construction itself. To obtain the desired quotient dimensioned Lie bracket we set:
\begin{equation}
    \{n_g+I_g,m_h+I_h\}':=\{n_g,m_h\}+I_{bgh}
\end{equation}
which is easily checked to be well-defined and that inherits the antisymmetry and Jacobi properties directly from dimensioned Lie bracket $\{,\}$ and the fact that $I\subset N(I)$ is a dimensioned Lie ideal.
\end{proof}

\begin{prop}[Heterogeneous Dimensioned Poisson Product] \label{DimensionedPoissonProductHetero}
Let $(A_G,*_g,\{,\}_g)$ and $(B_H,*_h,\{,\}_h)$ be dimensioned Poisson algebras over a dimensioned ring $R_P$, then the tensor product $A_G \otimes_{R_P} B_H$ carries a natural Poisson structure defined by linearly extending the following commutative product and Lie bracket:
\begin{align}
    (a\otimes b) *_{gh} (a' \otimes b) &:= a *_g a' \otimes b *_h b'\\
    \{a \otimes b, a' \otimes b'\}_{gh} &:= \{ a, a'\}_g \otimes b *_h b' + a*_g a'\otimes \{b,b'\}_h.
\end{align}
\end{prop}
\begin{proof}
From an algebraic point of view, this construction is entirely analogous to the ordinary product of Poisson algebras (see for instance \cite{fernandes2014lectures}), which is in fact explicitly recovered for trivial dimension monoids $G=H=\{1\}$. Then, it suffices to show that the two proposed expressions for are well-defined dimensioned algebras. First note that the tensor product of dimensioned monoids carries a natural monoid structure defined by:
\begin{equation}
    (g \otimes_P h ) \cdot (g' \otimes_P h') := gg' \otimes_P hh'.
\end{equation}
Taking the dimension projection of the commutative product we obtain:
\begin{equation}
    \delta(a_x *_g a_y' \otimes_{R_P} b_u *_h b_v') = gxy \otimes_P huv = ( g \otimes_P h ) \cdot (x \otimes_P u) \cdot (y \otimes_P v),
\end{equation}
hence $*_{gh}$ is a well-defined commutative multiplication of dimension $g \otimes_P h\in G \times^P H$. For the bracket $\{\,,\,\}_{gh}$ to be well-defined both terms need to have equal dimension projection, this is indeed the case since the dimension of the bracket and commutative products are assumed to be equal:
\begin{equation}
    \delta (\{ a_x, a_y'\}_g \otimes b_u *_h b_v') = (g \otimes_P h)\cdot (x \otimes_P u) \cdot (y \otimes_P v) = \delta (a_x*_g a_y'\otimes \{b_u,b_v'\}_h).
\end{equation}
\end{proof}

\begin{prop}[Homogeneous Dimensioned Poisson Product] \label{DimensionedPoissonProductHomo}
Let $(A_G,*_p,\{,\}_b)$ and $(B_G,*_q,\{,\}_c)$ be dimensioned Poisson algebras over a dimensioned ring $R_G$ such that $bq=pc$, then the tensor product $A_G \otimes_{R_G} B_G$ carries a natural Poisson structure defined by linearly extending the following commutative product and Lie bracket:
\begin{align}
    (a\otimes b) *_{pq} (a' \otimes b) &:= a *_p a' \otimes b *_q b'\\
    \{a \otimes b, a' \otimes b'\}_{pc} &:= \{ a, a'\}_b \otimes b *_q b' + a*_p a'\otimes \{b,b'\}_c.
\end{align}
\end{prop}
\begin{proof}
We proceed analogously to the proof of Proposition \ref{DimensionedPoissonProductHetero} above with the added simplicity that now all dimension sets are $G$ and the element-wise tensor product is given simply by monoid multiplication since $G\times^G G \cong G$. Taking the dimension projection of the commutative product we obtain:
\begin{equation}
    \delta(a_x *_p a_y' \otimes_{R_G} b_u *_q b_v') = pxy \otimes_G quv = pq xyuv = pq \delta(a_x \otimes_{R_G} b_u ) \delta(a_y' \otimes_{R_G} b_v'),
\end{equation}
hence $*_{pq}$ is a well-defined commutative multiplication of dimension $pq\in G $. For the bracket $\{\,,\,\}_{pc}$ to be well-defined both terms need to have equal dimension projection, this is indeed the case from the condition $bq=pc$ assumed of the pair of dimensioned Poisson algebras:
\begin{equation}
    \delta (\{ a_x, a_y'\}_b \otimes b_u *_q b_v') = bq xyuv = pc xyuv =  \delta (a_x*_p a_y'\otimes \{b_u,b_v'\}_c).
\end{equation}
\end{proof}

\section{Comments and Further Research} \label{conclusion}

The present work is part of a series of papers written by the author on the topic of Jacobi geometry and its generalizations. Jacobi structures on line bundles \cite{vitagliano2015dirac} \cite{tortorella2017deformations} \cite{schnitzer2019thesis} are a natural source of examples of dimensioned Poisson manifolds. In \cite{zapata2020poly} the author sets out to formulate the conventional theory of Jacobi manifolds employing the dimensioned algebra machinery that was developed in the present work as well as to explore natural generalizations suggested by the power functor of Section \ref{power} applied to collections of line bundles.\newline

The prequel paper to the present one \cite{zapata2020unitfree} introduces the notion of line-vector spaces as a unit-free analogue of a vector space. A natural question arises: is the notion of line-vector space related to that of a dimensioned vector space over powers of lines? More generally, what is the relation between the theory of line-vector spaces and dimensioned linear algebra? These questions will be answered in the sequel paper \cite{zapata2020poly}.\newline

We have focused exclusively on the algebraic aspects of dimensioned structures, but what about geometry? If one wishes to replicate ordinary manifold theory but replacing any instance of a coordinate space by its dimensioned analogue, it is possible that genuinely new mathematical objects may be found. Some comments in this direction are given in \cite[Ch. 8 ]{zapata2019landscape}.\newline

Returning to the motivating question about physical quantities: although the dimensioned formalism has been shown to fully account for standard dimensional analysis in general, ordinary physical quantities are always formulated with the particular groups $\Int$ or $\Rat$ as dimension sets. Besides the fact that a very prominent class of examples, i.e. the power dimensioned rings of lines, shares with this choice of dimension set, no further justification has been found to single out $\Int$ or $\Rat$ as canonical choices for dimension sets. Understanding the implications of this choice seems an important angle to further delve into the mathematical foundations of metrology. As a particularly enticing thought: could it be scientifically justifiable or theoretically useful to consider physical theories whose observable quantities have dimensions displaying some exotic algebraic or topological properties?\newline

Lastly, we mention recent work on graded geometry that uses definitions that very closely resemble those of our dimensioned sets and binars. In \cite[Sec. 1]{vysoky2021global} the definitions of \emph{graded sets}, \emph{graded morphisms} and \emph{graded abelian groups} match our definitions of dimensioned set, dimensioned map and dimensioned abelian group respectively by replacing all our generic dimension sets with $\Int$. Interestingly, the \emph{graded tensor product} of graded abelian groups \cite[Sec. 1.1, eq. (3)]{vysoky2021global}, which makes use of the additive structure of the base set $\Int$, differs significantly from our general definition of dimensioned tensor product. Understanding the differences and similarities between the dimensioned and graded formalisms seems like a fruitful line of further research.

\printbibliography

\end{document}